\documentclass[preprint,12pt,authoryear]{elsarticle}

\usepackage{amsmath,amssymb,amsthm}
\usepackage{mathtools}
\usepackage{graphicx}
\usepackage{booktabs}
\usepackage{tikz}
\usetikzlibrary{decorations.pathreplacing,decorations.pathmorphing}
\usepackage{pifont}
\usepackage{float}
\usepackage{hyperref}
\usepackage{cleveref}
\usepackage{lineno}
\newtheorem{theorem}{Theorem}[section]
\newtheorem{corollary}[theorem]{Corollary}
\newtheorem{lemma}[theorem]{Lemma}
\newtheorem{proposition}[theorem]{Proposition}

\theoremstyle{definition}
\newtheorem{definition}[theorem]{Definition}
\theoremstyle{remark}
\newtheorem{remark}[theorem]{Remark}
\newtheorem{example}[theorem]{Example}

\journal{}

\begin{document}

\begin{frontmatter}



\title{Geometric Characterization of Context-Free Intersections via the Inner Segment Dichotomy} 

\author{Jorge Miguel Silva\corref{cor1}}
\ead{jorge.miguel.ferreira.silva@ua.pt}
\cortext[cor1]{Corresponding author}

\affiliation{organization={IEETA/DETI, University of Aveiro},
            addressline={Campus Universitário de Santiago}, 
            city={Aveiro},
            postcode={3810-193}, 
            state={},
            country={Portugal}}

\begin{abstract}
The intersection of two context-free languages is not generally context-free, but no geometric criterion has characterized \emph{when} it remains so. The \emph{crossing gap} ($\max(i'-i,\;j'-j)$ for two crossing push-pop arcs) is the natural candidate. We refute this: we exhibit CFLs whose intersection is CFL despite unbounded-gap crossings.

The governing quantity is the \emph{inner segment measure}: for crossing arcs inducing a decomposition $w = P_1 P_2 P_3 P_4$, it is $\max(|P_2|,|P_3|)$, the length of the longer inner segment between interleaved crossing endpoints. We prove a dichotomy for this measure: bounded inner segments imply context-freeness via a finite buffer construction; growing inner segments with pump-sensitive linkages imply non-context-freeness. The inner segment concept applies to all CFL intersections; the strictness of the resulting characterization depends on the language class.

For block-counting CFLs (languages requiring equality among designated pairs of block lengths), the dichotomy is complete: the intersection is CFL if and only if the combined arcs are jointly well-nested. For general CFLs, the CFL direction is unconditional; the non-CFL direction requires pump-sensitive linkages whose necessity is the main open problem, reducing the general CFL intersection problem to a specific property of pump-sensitive decompositions.

\end{abstract}


\begin{highlights}
\item The crossing gap between PDA arcs does not determine context-freeness of CFL intersections.
\item The inner segment measure, the length of the longer inner segment between interleaved crossing endpoints, is the governing quantity.
\item Bounded inner segments yield a buffered PDA construction; growing inner segments with pump-sensitive linkages imply non-CFL.
\item For block-counting CFLs, the dichotomy is complete: CFL if and only if arcs are jointly well-nested.
\item The general CFL intersection problem reduces to a specific question about pump-sensitive decompositions.
\end{highlights}

%% Keywords

\end{frontmatter}



%% ====================================================================
\section{Introduction}\label{sec:intro}
%% ====================================================================

Context-free languages are not closed under intersection. The power of a pushdown automaton lies in its stack, and the stack's limitation is its rigidity: it provides last-in, first-out (LIFO) access, so information buried below the top is inaccessible until everything above it has been removed.

When two CFLs are intersected, the question is whether two independent stack protocols, each tracking its own push-pop dependencies, can be interleaved into a single stack. If the two sets of dependency arcs nest cleanly, one stack handles both. If they do not, the intersection may leave the CFL class entirely. The language $\{a^m b^n c^m d^n\}$, arising from $\{a^m b^* c^m d^*\} \cap \{a^* b^n c^* d^n\}$, is the classic example: its $a$-to-$c$ and $b$-to-$d$ dependencies \emph{cross}, and no single stack can process both~\cite{shieber1985,joshi1990}. Bodirsky, Kuhlmann, and M\"ohl~\cite{bodirsky2005} formalized this via \emph{well-nestedness} in the LCFRS hierarchy.

However, not all crossings destroy context-freeness. If arcs from the two PDAs cross only at adjacent positions, as in the interleaved palindrome language of \Cref{ex:palindrome}, the crossing is local and a pair-alphabet encoding absorbs it. This raises a geometric question: \emph{what property of crossing dependency arcs determines whether two stack protocols can share one stack?} While initially motivated by structural decomposition problems in sequence analysis and data compression, this question applies to all context-free languages and has remained open in the automata-theoretic literature.

\medskip
\noindent\textbf{The natural conjecture, and its refutation.} The \emph{crossing gap} of two crossing arcs, defined formally as $\max(i'-i,\;j'-j)$ (the larger of the push-endpoint distance and the pop-endpoint distance; \Cref{def:gap}), is a natural candidate: if the gap grows with input size, the stack should break; bounded-gap crossings should be harmless.

This conjecture is false. In \Cref{prop:gap-refuted}, we exhibit two CFLs whose intersection remains context-free despite having crossing arcs with crossing gap $\Theta(n)$. The construction is explicit: one PDA performs a single push-pop pair spanning two positions, while the other PDA's matching arc spans $\Theta(n)$ positions around it. Both inner segments have length exactly~$1$, so the ``interruption'' is absorbed by a finite buffer, yet the crossing gap grows without bound. Intuitively, an arc that spans many positions but encloses no other machine's operations creates a large gap without imposing any burden on the stack.

\medskip
\noindent\textbf{The governing measure: inner segments.} The right quantity is the \emph{inner segment measure} $\max(|P_2|,|P_3|)$. When two arcs $(i,j)$ and $(i',j')$ cross with $i < i' < j < j'$, the string decomposes as $w = P_1 P_2 P_3 P_4$, and this measure captures the length of the longer inner segment between the interleaved endpoints. Bounded interruptions are \emph{bufferable}: a finite-state buffer absorbs the short crossing arcs, and a single PDA handles the rest. Unbounded interruptions break the LIFO discipline irreparably.

Our main results formalize this intuition:

\begin{enumerate}
\item \textbf{Bounded inner segments $\Rightarrow$ CFL} (\Cref{thm:buffer}): if all crossing pairs have $\max(|P_2|,|P_3|) = O(1)$, we construct a PDA recognizing the intersection. The buffer size is exponential in the bound but independent of input length.

\item \textbf{Growing inner segments $\Rightarrow$ not CFL} (\Cref{thm:strengthened}): if some crossing pair has $\max(|P_2|,|P_3|) = \omega(1)$ with pump-sensitive linkages (dependencies where modifying one linked segment without its partner breaks membership), the intersection is not context-free.

\item \textbf{Complete characterization for block-counting CFLs} (\Cref{thm:characterization}): the intersection is CFL if and only if the combined arcs are jointly well\-nested.

\item \textbf{Refutation of the gap conjecture} (\Cref{prop:gap-refuted}): unbounded crossing gap does not imply non-CFL; inner segment measure is the governing quantity.
\end{enumerate}

\noindent Results (1) and (2) combine into the \textbf{inner segment dichotomy} (\Cref{cor:dichotomy}), summarized in the following table:

\begin{center}
\resizebox{\textwidth}{!}{%
\begin{tabular}{llll}
\toprule
Inner segment measure & Crossing gap & Intersection & Reference \\
\midrule
$0$ (no crossings) & -- & CFL & (product construction) \\
any & $O(1)$ & CFL & \Cref{thm:bounded-gap} \\
$O(1)$ & $\omega(1)$ & CFL & \Cref{thm:buffer} \\
$\omega(1)$$^\dagger$ & $\omega(1)$ & Not CFL & \Cref{thm:strengthened} \\
\bottomrule
\end{tabular}}
\end{center}
{\small $^\dagger$Requires pump-sensitive linkages; automatic for block-counting CFLs, open for general CFLs (\Cref{rem:exhaustiveness}). Note: bounded gap (row~2) always yields CFL, so the linkage condition is satisfiable only when the gap is also unbounded.}

\medskip

\noindent\textbf{Scope.} \Cref{sec:pda} fixes notation. \Cref{sec:theorem4} through \Cref{sec:intermediate} develop the formal results. \Cref{sec:related} surveys related work, \Cref{sec:discussion} places the results in the Chomsky hierarchy, and \Cref{sec:open} lists open problems.

%% ====================================================================
\section{Preliminaries}\label{sec:pda}
%% ====================================================================

We work with nondeterministic pushdown automata in Greibach normal form (GNF): each transition reads exactly one input symbol and performs at most one stack operation, either a single push or a single pop. Binary productions ($A\to aBC$) require pushing two symbols; we decompose each such step into two transitions via an auxiliary $\epsilon$-move. Accordingly, each input position is associated with at most two pushes per machine. The details of this decomposition appear in \Cref{sec:bounded-gap}.

For a PDA~$M$ accepting a string~$w$, the \emph{push-pop matching} $\mu_M(w)$ is the set of arcs $(i,j)$ pairing the position~$i$ of a push with the position~$j$ of the corresponding pop. Since~$M$ is nondeterministic, different accepting computations may produce different matchings. Our results require only that \emph{some} accepting computation of each machine yields the relevant crossing structure; the specific computation is fixed once and used throughout.

Two arcs $(i,j)\in\mu_{M_1}(w)$ and $(i',j')\in\mu_{M_2}(w)$ \emph{cross} if $i<i'<j<j'$ or $i'<i<j'<j$. (By relabeling the machines, any crossing pair can be written with $i<i'<j<j'$; we adopt this convention throughout.) Such a crossing induces \emph{inner segments} between the interleaved endpoints; the formal decomposition appears in \Cref{sec:crossing-geometry}.

\begin{example}\label{ex:palindrome}
The interleaved palindrome language
\[
  L = \{w\in\{0,1\}^{2n} : w_1 w_3\cdots w_{2n-1}\in\mathrm{Pal},\;w_2 w_4\cdots w_{2n}\in\mathrm{Pal}\}
\]
is context-free: defining pair symbols $X_i=(w_{2i-1},w_{2i})$, both palindrome conditions require $X_i=X_{n-i+1}$, so $L$ is a palindrome language over $\{00,01,10,11\}$. Both regularities share the same matching structure (position~$i$ linked to~$n+1-i$), and the crossing gap is~$1$, i.e., adjacent positions. The pair-alphabet encoding absorbs the crossing.
\end{example}

%% ====================================================================
\section{Crossing pump-sensitive linkages imply non-CFL}\label{sec:theorem4}
%% ====================================================================

\subsection{Pump-sensitive linkage}

The non-CFL direction of our results requires a coupling condition on string segments. Informally, two segments of a string are \emph{pump-sensitively linked} if they are so tightly coupled that modifying one without the other breaks language membership. Unlike the standard pumping lemma, which guarantees that \emph{some} factorization can be pumped \emph{while remaining in the language}, pump-sensitive linkage requires the opposite: \emph{every} isolated pump (any factorization that touches one linked segment while avoiding the other) takes the string \emph{out of} the language.

\begin{definition}\label{def:linkage}
Let $L\subseteq\Sigma^*$. A \emph{pump-sensitive linkage} $(P_i,P_j)$ in~$L$ on a string $w=P_1 P_2\cdots P_k$ (where $P_1,\dots,P_k$ are consecutive segments and $i<j$) is the property that: for any factorization $w=uvxyz$ with $|vy|\ge 1$ and $vxy$ intersecting~$P_i$ but disjoint from~$P_j$, or intersecting~$P_j$ but disjoint from~$P_i$, we have $uv^2xy^2z\notin L$.
\end{definition}

Informally: pumping anywhere in one linked segment, without simultaneously modifying the other, takes the string out of~$L$. As a familiar example, in the language $\{a^n b^n\}$, the segments $a^n$ and $b^n$ are pump-sensitively linked: any factorization that pumps characters in the $a$-region alone (without touching the $b$-region) breaks the equality $\#_a = \#_b$, and vice versa. The definition generalizes this tight coupling from simple counting constraints to arbitrary language membership conditions. Note that linkage is an \emph{input condition} verified from the language's membership conditions, analogous to choosing marked positions in Ogden's lemma, not a conclusion derived from knowing the language is non-CFL. For block-counting CFLs, it follows directly from the equality constraints (\Cref{sec:block-inner}); for general CFLs, whether it can always be derived from arc geometry alone is the main open problem (\Cref{sec:open}, Question~1).

\begin{remark}[Universal quantification is intentional]\label{rem:linkage-quantifier}
The definition requires that \emph{every} factorization touching one linked segment alone fails to pump, not merely the existence of such a factorization. This is strictly stronger than what the pumping lemma provides (which guarantees only \emph{some} pumpable factorization). The strength is necessary: \Cref{thm:crossing} works by showing that the pumping lemma's guaranteed factorization must fall in one of seven cases, each of which touches exactly one member of some linked pair. For the argument to go through, \emph{all} such cases must fail; hence the universal quantification.

Verifying this condition for specific languages requires showing that the segments are ``tightly coupled'' to their partners: any local modification to one side, however positioned, breaks the structural constraint. For block-counting CFLs with crossing arcs, this holds on the intersection~$L_1\cap L_2$: any pump modifying one block's count without modifying its constrained partner's violates at least one of the intersection's equality constraints (see \Cref{rem:linkage-ambient}). For general CFLs, verifying the condition is nontrivial and is the source of the conditionality in the general dichotomy.
\end{remark}

\subsection{The crossing linkage theorem}

\begin{theorem}[Pumping incompatibility of crossing linkages]\label{thm:crossing}
Let $L\subseteq\Sigma^*$ be a language containing, for infinitely many~$n$, a string~$w_n$ with a factorization $w_n=P_1 P_2 P_3 P_4$ satisfying:
\begin{enumerate}
\item[\textup{(i)}] $|P_i|\ge n$ for each $i\in\{1,2,3,4\}$.
\item[\textup{(ii)}] $L$ has a pump-sensitive linkage $(P_1,P_3)$.
\item[\textup{(iii)}] $L$ has a pump-sensitive linkage $(P_2,P_4)$.
\end{enumerate}
Then $L$ is not context-free.
\end{theorem}

\Cref{fig:pumping} illustrates the case analysis: every placement of the bounded-length substring~$vxy$ touches exactly one member of a linked pair, triggering one of the two linkage conditions.

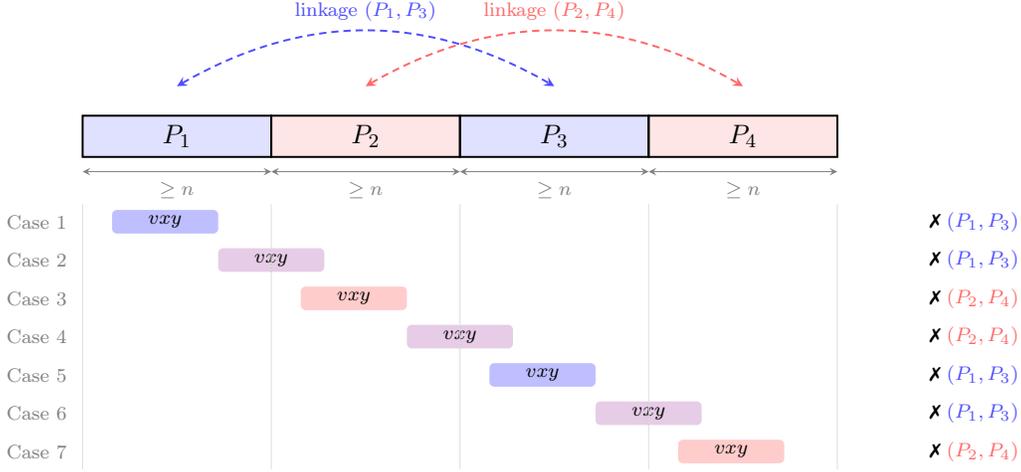
\begin{figure}[ht]
\centering
\resizebox{\textwidth}{!}{%
\begin{tikzpicture}[>=stealth, scale=0.85]
  \def\segH{0.7}
  \def\sw{3.2} % segment width

  % Draw segments
  \fill[blue!12] (0,0) rectangle (\sw,\segH);
  \fill[red!10]  (\sw,0) rectangle (2*\sw,\segH);
  \fill[blue!12] (2*\sw,0) rectangle (3*\sw,\segH);
  \fill[red!10]  (3*\sw,0) rectangle (4*\sw,\segH);

  \draw[thick] (0,0) rectangle (\sw,\segH);
  \draw[thick] (\sw,0) rectangle (2*\sw,\segH);
  \draw[thick] (2*\sw,0) rectangle (3*\sw,\segH);
  \draw[thick] (3*\sw,0) rectangle (4*\sw,\segH);

  % Segment labels
  \node[font=\small] at (0.5*\sw,0.35) {$P_1$};
  \node[font=\small] at (1.5*\sw,0.35) {$P_2$};
  \node[font=\small] at (2.5*\sw,0.35) {$P_3$};
  \node[font=\small] at (3.5*\sw,0.35) {$P_4$};

  % Size annotations below main bar
  \foreach \i in {0,1,2,3} {
    \draw[<->,gray,font=\scriptsize] ({\i*\sw},-0.25) -- node[below]{$\ge n$} ({(\i+1)*\sw},-0.25);
  }

  % Linkage arcs above — both at same height
  \draw[red!60,thick,<->,densely dashed] (1.5*\sw,1.2) to[out=30,in=150]
    node[above,font=\scriptsize]{linkage $(P_2,P_4)$} (3.5*\sw,1.2);

  \draw[blue!70,thick,<->,densely dashed] (0.5*\sw,1.2) to[out=30,in=150]
    node[above,font=\scriptsize]{linkage $(P_1,P_3)$} (2.5*\sw,1.2);

  % === Case rows below ===
  % Each case on its own row, with generous vertical spacing
  \def\rh{0.4}  % row height
  \def\rs{0.65} % row spacing
  \def\bw{1.8}  % vxy box width
  \def\labelX{14.2} % right-aligned label position

  % Row template: y-coord, x-start of vxy box, fill color, case label, contradiction
  % Case 1: vxy ⊆ P1
  \pgfmathsetmacro{\ry}{-1.3}
  \fill[blue!25,rounded corners=2pt] (0.5,\ry) rectangle ({0.5+\bw},{\ry+\rh});
  \node[font=\scriptsize] at ({0.5+0.5*\bw},{\ry+0.5*\rh}) {$vxy$};
  \node[font=\scriptsize,anchor=west] at (\labelX,{\ry+0.5*\rh}) {\ding{55}\;\textcolor{blue!70}{$(P_1,P_3)$}};
  \node[font=\scriptsize,gray,anchor=east] at (-0.1,{\ry+0.5*\rh}) {Case 1};

  % Case 2: vxy straddles P1–P2
  \pgfmathsetmacro{\ry}{-1.3 - \rs}
  \fill[violet!20,rounded corners=2pt] ({0.5*\sw+0.7},\ry) rectangle ({0.5*\sw+0.7+\bw},{\ry+\rh});
  \draw[gray,densely dotted] (\sw,\ry) -- (\sw,{\ry+\rh});
  \node[font=\scriptsize] at ({0.5*\sw+0.7+0.5*\bw},{\ry+0.5*\rh}) {$vxy$};
  \node[font=\scriptsize,anchor=west] at (\labelX,{\ry+0.5*\rh}) {\ding{55}\;\textcolor{blue!70}{$(P_1,P_3)$}};
  \node[font=\scriptsize,gray,anchor=east] at (-0.1,{\ry+0.5*\rh}) {Case 2};

  % Case 3: vxy ⊆ P2
  \pgfmathsetmacro{\ry}{-1.3 - 2*\rs}
  \fill[red!20,rounded corners=2pt] ({\sw+0.5},\ry) rectangle ({\sw+0.5+\bw},{\ry+\rh});
  \node[font=\scriptsize] at ({\sw+0.5+0.5*\bw},{\ry+0.5*\rh}) {$vxy$};
  \node[font=\scriptsize,anchor=west] at (\labelX,{\ry+0.5*\rh}) {\ding{55}\;\textcolor{red!60}{$(P_2,P_4)$}};
  \node[font=\scriptsize,gray,anchor=east] at (-0.1,{\ry+0.5*\rh}) {Case 3};

  % Case 4: vxy straddles P2–P3
  \pgfmathsetmacro{\ry}{-1.3 - 3*\rs}
  \fill[violet!20,rounded corners=2pt] ({1.5*\sw+0.7},\ry) rectangle ({1.5*\sw+0.7+\bw},{\ry+\rh});
  \draw[gray,densely dotted] ({2*\sw},\ry) -- ({2*\sw},{\ry+\rh});
  \node[font=\scriptsize] at ({1.5*\sw+0.7+0.5*\bw},{\ry+0.5*\rh}) {$vxy$};
  \node[font=\scriptsize,anchor=west] at (\labelX,{\ry+0.5*\rh}) {\ding{55}\;\textcolor{red!60}{$(P_2,P_4)$}};
  \node[font=\scriptsize,gray,anchor=east] at (-0.1,{\ry+0.5*\rh}) {Case 4};

  % Case 5: vxy ⊆ P3
  \pgfmathsetmacro{\ry}{-1.3 - 4*\rs}
  \fill[blue!25,rounded corners=2pt] ({2*\sw+0.5},\ry) rectangle ({2*\sw+0.5+\bw},{\ry+\rh});
  \node[font=\scriptsize] at ({2*\sw+0.5+0.5*\bw},{\ry+0.5*\rh}) {$vxy$};
  \node[font=\scriptsize,anchor=west] at (\labelX,{\ry+0.5*\rh}) {\ding{55}\;\textcolor{blue!70}{$(P_1,P_3)$}};
  \node[font=\scriptsize,gray,anchor=east] at (-0.1,{\ry+0.5*\rh}) {Case 5};

  % Case 6: vxy straddles P3–P4
  \pgfmathsetmacro{\ry}{-1.3 - 5*\rs}
  \fill[violet!20,rounded corners=2pt] ({2.5*\sw+0.7},\ry) rectangle ({2.5*\sw+0.7+\bw},{\ry+\rh});
  \draw[gray,densely dotted] ({3*\sw},\ry) -- ({3*\sw},{\ry+\rh});
  \node[font=\scriptsize] at ({2.5*\sw+0.7+0.5*\bw},{\ry+0.5*\rh}) {$vxy$};
  \node[font=\scriptsize,anchor=west] at (\labelX,{\ry+0.5*\rh}) {\ding{55}\;\textcolor{blue!70}{$(P_1,P_3)$}};
  \node[font=\scriptsize,gray,anchor=east] at (-0.1,{\ry+0.5*\rh}) {Case 6};

  % Case 7: vxy ⊆ P4
  \pgfmathsetmacro{\ry}{-1.3 - 6*\rs}
  \fill[red!20,rounded corners=2pt] ({3*\sw+0.5},\ry) rectangle ({3*\sw+0.5+\bw},{\ry+\rh});
  \node[font=\scriptsize] at ({3*\sw+0.5+0.5*\bw},{\ry+0.5*\rh}) {$vxy$};
  \node[font=\scriptsize,anchor=west] at (\labelX,{\ry+0.5*\rh}) {\ding{55}\;\textcolor{red!60}{$(P_2,P_4)$}};
  \node[font=\scriptsize,gray,anchor=east] at (-0.1,{\ry+0.5*\rh}) {Case 7};

  % Light reference lines from main bar segments down through cases
  \foreach \i in {0,1,2,3,4} {
    \draw[gray!25] ({\i*\sw},-0.8) -- ({\i*\sw},-1.3-6*\rs-0.1);
  }
\end{tikzpicture}%
}% end resizebox
\caption{The pumping obstruction for crossing linkages (\Cref{thm:crossing}). \emph{Top:} the string $w_n = P_1 P_2 P_3 P_4$ with crossing linkages $(P_1,P_3)$ and $(P_2,P_4)$. \emph{Bottom:} since $|vxy| \le p < n \le |P_i|$, the pumpable substring (shown for each case) intersects at most two adjacent segments. Every placement triggers a linkage contradiction (\ding{55}): it touches exactly one member of a linked pair without reaching the other.}
\label{fig:pumping}
\end{figure}

\begin{proof}
Assume for contradiction that $L$ is CFL with pumping length~$p$. Choose $n>p$ and consider $w_n=P_1 P_2 P_3 P_4$ with $|P_i|\ge n>p$.

By the pumping lemma, there exists a factorization $w_n=uvxyz$ with $|vxy|\le p$, $|vy|\ge 1$, and $uv^ixy^iz\in L$ for all $i\ge 0$. Since $|vxy|\le p<n\le|P_i|$, the substring~$vxy$ can intersect at most two adjacent segments. We enumerate all possible locations:

\emph{Case~1: $vxy\subseteq P_1$.} Then $vxy$ is disjoint from~$P_3$. By~(ii), $uv^2xy^2z\notin L$. Contradiction.

\emph{Case~2: $vxy$ straddles the $P_1 P_2$ boundary.} Then $vxy$ is disjoint from both $P_3$ and~$P_4$. Since $vxy$ intersects~$P_1$ but not~$P_3$, condition~(ii) gives a contradiction.

\emph{Case~3: $vxy\subseteq P_2$.} Disjoint from~$P_4$. By~(iii), contradiction.

\emph{Case~4: $vxy$ straddles the $P_2 P_3$ boundary.} Since $|P_2|>p\ge|vxy|$, the start of~$vxy$ lies strictly inside~$P_2$, so $vxy$ is disjoint from~$P_1$. Since $|P_3|>p\ge|vxy|$, the end of~$vxy$ lies strictly inside~$P_3$, so $vxy$ is disjoint from~$P_4$. Hence $vxy$ intersects~$P_2$ but not~$P_4$; by~(iii), contradiction.

\emph{Cases~5--7} ($vxy\subseteq P_3$, straddling $P_3 P_4$, $vxy\subseteq P_4$) are symmetric. All cases yield contradictions.
\end{proof}

\subsection{Application to CFL intersection}

\begin{corollary}\label{cor:cfl}
Let $L_1,L_2$ be context-free languages. Suppose there exist, for infinitely many~$n$, a string $w_n\in L_1\cap L_2$ and a factorization $w_n=P_1 P_2 P_3 P_4$ with $|P_i|\ge n$, such that $L_1$ has a pump-sensitive linkage $(P_1,P_3)$ and $L_2$ has a pump-sensitive linkage $(P_2,P_4)$. Then $L_1\cap L_2$ is not context-free.
\end{corollary}

\begin{proof}
Since $L_1\cap L_2\subseteq L_1$, pump-sensitive linkages of~$L_1$ are inherited by $L_1\cap L_2$. Similarly for~$L_2$. Apply \Cref{thm:crossing}.
\end{proof}

\subsection{Verification on known examples}

\begin{example}[$\{a^n b^n c^n d^n\}$]\label{ex:abcd}
Let $L_1=\{w\in a^*b^*c^*d^*:\#_a(w)=\#_c(w)\}$, $L_2=\{w\in a^*b^*c^*d^*:\#_b(w)=\#_d(w)\}$, both context-free. Set $L=L_1\cap L_2=\{a^n b^n c^n d^n\}$. Factor $a^n b^n c^n d^n$ as $P_1=a^n$, $P_2=b^n$, $P_3=c^n$, $P_4=d^n$, all of size~$n$.

We verify the pump-sensitive linkages on~$L$ itself (not on~$L_1$ or~$L_2$ individually; see \Cref{rem:linkage-ambient} below). Linkage $(P_1,P_3)$ on~$L$: any factorization $w=uvxyz$ where $vxy$ intersects~$P_1$ but is disjoint from~$P_3$ has $vxy$ contained in the $a$-$b$ region. Since $|vy|\ge 1$, pumping changes at least one of~$\#_a$ or~$\#_b$: if~$\#_a$ changes, then $\#_a\neq\#_c=n$ violates the $L_1$-constraint; if~$\#_a$ is unchanged, then some $b$'s were pumped, giving $\#_b\neq n=\#_d$, violating the $L_2$-constraint. Either way, $uv^2xy^2z\notin L$. Linkage $(P_2,P_4)$ on~$L$: symmetric. By \Cref{thm:crossing}, $L$ is not CFL.
\end{example}

\begin{remark}\label{rem:linkage-ambient}
The linkages above hold on $L=L_1\cap L_2$, not on~$L_1$ alone. On~$L_1=\{w:\#_a=\#_c\}$, the linkage $(P_1,P_3)$ fails: a cross-boundary factorization with $v=\epsilon$, $x=a\cdot b^s$, $y=b^t$ has $vxy$ intersecting~$P_1$ but disjoint from~$P_3$, yet pumping just adds~$b$'s, preserving $\#_a=\#_c$. The intersection's additional constraints ($\#_b=\#_d$) prevent this escape.

This subtlety means \Cref{cor:cfl}, which requires linkages on the ambient languages~$L_1,L_2$ individually, does not apply to this example. Instead, we apply \Cref{thm:crossing} directly to~$L$. For block-counting CFLs, linkages on the intersection always hold: any pump modifying one block's count, without modifying its partner's, violates one of the intersection's constraints.
\end{remark}

\begin{example}[$\{a^n b^n c^n\}$]\label{ex:abc}
Take $L_1=\{a^n b^n w:w\in c^*\}$ and $L_2=\{w b^n c^n:w\in a^*\}$, both CFL. The block-counting arcs are $(1,2)$ from~$L_1$ and $(2,3)$ from~$L_2$, which share endpoint~$2$; accordingly, $L_1\cap L_2=\{a^n b^n c^n\}$ enforces the three-way equality $|B_1|=|B_2|=|B_3|$, which is not CFL by the standard pumping argument~\cite{hopcroft2007}. (This illustrates the shared-endpoint mechanism of \Cref{def:joint-wn}(iv) rather than crossing arcs; the full characterization is given in \Cref{thm:characterization}.)
\end{example}

\subsection{What the crossing theorem covers}

\Cref{thm:crossing} handles any pair of CFLs whose intersection admits strings with four linearly-growing segments and crossing pump-sensitive linkages, including all block-counting CFLs with crossing arcs. \Cref{thm:strengthened} (in \Cref{sec:intermediate}) relaxes the requirement from all four segments to a single inner segment growing unboundedly.

\Cref{thm:crossing} does not cover bounded-gap crossings. The condition $|P_i|\ge n$ ensures all four segments grow with~$n$. \Cref{thm:bounded-gap} proves that bounded-gap crossings ($O(1)$) are always resolvable.

%% ====================================================================
\section{Complete characterization for block-counting CFLs}\label{sec:characterization}
%% ====================================================================

\subsection{Block-counting CFLs}

\begin{definition}\label{def:block-counting}
A \emph{block-counting CFL} over disjoint sub-alphabets $\Sigma_1,\dots,\Sigma_k$ is a language of strings $w=B_1 B_2\cdots B_k$ with $B_i\in\Sigma_i^*$, subject to equality constraints $C\subseteq\binom{[k]}{2}$ on block lengths: for each $(i,j)\in C$, $|B_i|=|B_j|$. Each constraint $(i,j)\in C$ is called an \emph{arc} from block~$i$ to block~$j$.
\end{definition}

\subsection{Joint well-nestedness}

\begin{definition}\label{def:well-nested}
An arc set $C$ on $\{1,\dots,k\}$ is \emph{well-nested} if no two arcs in~$C$ cross: there are no $(i,j),(i',j')\in C$ with $i<i'<j<j'$.
\end{definition}

When $C$ is well-nested, a block-counting language with arc set~$C$ is context-free: a PDA pushes during the left-endpoint block and pops during the right-endpoint block of each arc, with the LIFO discipline ensured by well-nestedness.

\begin{definition}\label{def:joint-wn}
Two arc sets $C_1,C_2$ on $\{1,\dots,k\}$ are \emph{jointly well-nested} if: (i)~$C_1$ is well-nested; (ii)~$C_2$ is well-nested; (iii)~no arc in~$C_1$ crosses an arc in~$C_2$: there are no $(i,j)\in C_1$ and $(i',j')\in C_2$ with $i<i'<j<j'$ or $i'<i<j'<j$; and (iv)~no \emph{distinct} arcs $(i,j)\in C_1$ and $(i',j')\in C_2$ with $(i,j)\neq(i',j')$ share an endpoint (i.e., $\{i,j\}\cap\{i',j'\}=\emptyset$). Identical arcs appearing in both~$C_1$ and~$C_2$ do not violate joint well-nestedness.
\end{definition}

\subsection{The characterization theorem}

\begin{theorem}[Characterization of block-counting CFL intersection]\label{thm:characterization}
Let $L_1,L_2$ be block-counting CFLs over the same block structure $(\Sigma_1,\dots,\Sigma_k)$ with well-nested arc sets $C_1,C_2$ respectively. Then
\[
  L_1\cap L_2\textup{ is context-free}\quad\Longleftrightarrow\quad C_1\textup{ and }C_2\textup{ are jointly well-nested.}
\]
\end{theorem}

\begin{proof}
\emph{($\Leftarrow$: Well-nested implies CFL.)} Suppose $C_1$ and $C_2$ are jointly well-nested. Construct a PDA~$M$ recognizing $L_1\cap L_2$. The stack alphabet is $\Gamma=\{c_e:e\in C_1\cup C_2\}\cup\{\$\}$. The PDA tracks the current block index via disjoint sub-alphabets. For each block~$B_\ell$ ($\ell\in\{1,\dots,k\}$): if~$\ell$ is the left endpoint of arc~$e=(\ell,r)\in C_1\cup C_2$, push~$c_e$ for each symbol read; if~$\ell$ is the right endpoint of arc~$e=(s,\ell)$, pop and verify~$c_e$; if~$\ell$ is free (not an endpoint of any arc), read without stack operations. Accept when only~$\$$ remains.

Joint well-nestedness ensures LIFO compatibility: when reaching the closing block~$B_\ell$ of arc~$e=(s,\ell)$, all arcs opened after~$e$ are either nested inside (closed before~$B_\ell$) or disjoint and later (not yet opened). No arc crosses~$e$ or shares its endpoint. Therefore $c_e$ markers are on top of the stack when needed.

\emph{($\Rightarrow$: Non-well-nested implies non-CFL.)} If joint well-nestedness fails, then either (a)~crossing arcs exist, or (b)~distinct arcs share an endpoint.

For~(a): say $(i,j)\in C_1$ and $(i',j')\in C_2$ cross; without loss of generality (relabeling $C_1,C_2$ if necessary), $i<i'<j<j'$. Assume for contradiction that $L=L_1\cap L_2$ is CFL with pumping length~$p$. Choose $n>p$. Call a block index~$\ell$ \emph{crossing-connected} if it is reachable from~$\{i,i',j,j'\}$ via arcs in~$C_1\cup C_2$ (treating arcs as undirected edges on block indices). Let $w\in L$ have $|B_\ell|=n$ for every crossing-connected block~$\ell$ and $|B_\ell|=0$ for all other blocks. All constraints are satisfied: constraints between two crossing-connected blocks have $n=n$; constraints between two non-crossing-connected blocks have $0=0$; and no arc connects a crossing-connected block to a non-crossing-connected one (by definition of reachability). The pumping lemma gives a factorization $uvxyz$ with $|vxy|\le p$ and $|vy|\ge 1$.

Since crossing-connected blocks have size $n>p$ and all other blocks are empty, $vxy$ lies within at most two adjacent crossing-connected blocks. Every such block is connected via a chain of arcs to some crossing endpoint~$\alpha\in\{i,i',j,j'\}$, whose crossing partner~$\beta$ satisfies $|\alpha-\beta|\ge 2$ (since $i<i'<j<j'$ with $i'$ between~$i$ and~$j$, and $j$ between~$i'$ and~$j'$); consequently $\beta$ is separated from~$\alpha$ by at least one intervening block of size~$n$, placing it at distance~$>p$ in the string. Since $|vxy|\le p$, any factorization touching a crossing-connected block cannot reach that block's crossing partner. If $vxy$ lies entirely within a single block~$B_\ell$, pumping changes~$|B_\ell|$ while the partner block is unreachable, violating a constraint. If $vxy$ straddles two adjacent blocks $B_\ell$ and~$B_{\ell+1}$, pumping changes at least one block's length; tracing that block's arc chain to a crossing endpoint yields a distant partner whose length remains~$n$, again violating a constraint. So $L$ is not CFL.

For~(b): distinct arcs from~$C_1$ and~$C_2$ share an endpoint: say $(i,j)\in C_1$ and $(i,j')\in C_2$ with $j\neq j'$. Then $L_1\cap L_2$ enforces $|B_i|=|B_j|$ and $|B_i|=|B_{j'}|$, i.e., $|B_i|=|B_j|=|B_{j'}|$, a three-way equality over disjoint alphabets, which is not CFL by a standard pumping argument (cf.~\cite{hopcroft2007}). (If $C_1$ and $C_2$ share an identical arc $(i,j)$, no new constraint is added, and the shared arc does not obstruct CFL-ness.)
\end{proof}

\subsection{Discussion}

\Cref{thm:characterization} provides a complete characterization for block-counting CFLs. \Cref{thm:strengthened} (in \Cref{sec:intermediate}) recovers this characterization as a consequence of the inner-segment-based analysis, since all crossings in the block-counting setting have growing inner segments with pump-sensitive linkages on~$L_1\cap L_2$ (cf.\ \Cref{rem:linkage-ambient}). Extending to general CFLs requires two ingredients: a grammar-independent definition of ``matching arcs,'' and a proof that non-CFL intersections always produce pump-sensitive linkages on some crossing pair. The CFL direction of the inner segment dichotomy (\Cref{thm:buffer}) is unconditional; the non-CFL direction (\Cref{thm:strengthened}) is conditional on the linkage assumption. The connection to well-nestedness in the LCFRS hierarchy is discussed in \Cref{sec:related}.

%% ====================================================================
\section{Bounded-gap crossings are always resolvable}\label{sec:bounded-gap}
%% ====================================================================

\subsection{Definitions}

\textbf{PDA normal form convention.} Throughout this section and \Cref{sec:intermediate}, we assume PDAs are in a normal form where each transition reads exactly one input symbol and performs at most one stack operation (push of a single symbol, or pop). Any PDA can be converted to this form: first convert the CFG to Greibach Normal Form (GNF), then construct the standard GNF-to-PDA translation. The resulting PDA may push up to two symbols per step (from binary productions); we decompose each such step into sequential single-push transitions via auxiliary states that perform $\epsilon$-moves.

The auxiliary states increase $|Q|$ by a grammar-dependent constant factor but do not read input symbols, so each \emph{input position} is associated with at most two pushes and one pop of a single symbol. Buffer bounds depending on arc span (measured in input positions) therefore correctly limit the number of simultaneously pending operations, up to the constant factor of~$2$.

Recall from \Cref{sec:pda} that $\mu_M(w)$ denotes the push-pop matching of PDA~$M$ on string~$w$.

\begin{definition}[Crossing gap]\label{def:gap}
For PDAs $M_1,M_2$ and string $w\in L(M_1)\cap L(M_2)$, a \emph{crossing pair} consists of arcs $(i,j)$ and $(i',j')$ from different machines (i.e., one from~$\mu_{M_1}(w)$ and the other from~$\mu_{M_2}(w)$) with $i<i'<j<j'$. The \emph{crossing gap} of this pair is $\max(i'-i,\;j'-j)$.
\end{definition}

\begin{definition}[Bounded crossing gap]\label{def:bounded-gap}
PDAs $M_1,M_2$ have \emph{$k$-bounded crossing gap} if for every $w\in L(M_1)\cap L(M_2)$, every crossing pair has crossing gap at most~$k$.
\end{definition}

Note that the gap must be bounded on \emph{both} sides for this construction. An arc pair with left endpoints close together ($i'-i = O(1)$) and inner overlap small ($j-i' = O(1)$) but right endpoints far apart ($j'-j = \omega(1)$) has unbounded crossing gap and cannot be resolved by \Cref{thm:bounded-gap} alone; however, the inner segment measure $\max(i'-i,\;j-i')$ remains bounded, and \Cref{thm:buffer} shows that such crossings are resolvable.

\subsection{Well-nestedness of non-crossing arc sets}

The displacement buffer construction below pushes all entries directly to the stack and, at pop time, temporarily displaces at most~$2k$ intervening entries from the other machine. The following elementary lemma, used later also in \Cref{thm:buffer}, justifies that non-crossing arcs from two machines can share a single stack.

\begin{lemma}[LIFO compatibility of non-crossing arcs]\label{lem:lifo}
Let $S_1$ and $S_2$ be well-nested arc sets (no two arcs within the same set cross). If no arc in~$S_1$ crosses any arc in~$S_2$, then $S_1\cup S_2$ is well-nested.
\end{lemma}

\begin{proof}
Suppose for contradiction that $(a,b)\in S_1\cup S_2$ and $(c,d)\in S_1\cup S_2$ cross, i.e., $a<c<b<d$. If both belong to~$S_1$ (or both to~$S_2$), this contradicts internal well-nestedness. If one belongs to~$S_1$ and the other to~$S_2$, this contradicts the no-crossing hypothesis. Therefore $S_1\cup S_2$ is well-nested.
\end{proof}

Since any well-nested arc set admits a LIFO ordering of its push/pop operations, \Cref{lem:lifo} guarantees that all non-crossing arcs from both machines can be processed by a single stack without reordering.

\subsection{The bounded-gap resolution theorem}

\begin{theorem}[Bounded-gap resolution]\label{thm:bounded-gap}
If PDAs $M_1,M_2$ have $k$-bounded crossing gap, then $L(M_1)\cap L(M_2)$ is context-free.
\end{theorem}

\begin{proof}
We construct a nondeterministic PDA~$M$ recognizing $L(M_1)\cap L(M_2)$. The key idea is to push all symbols directly to the stack with owner tags, and use a bounded \emph{displacement buffer} at pop time to skip over at most~$2k$ entries from the other machine.

\emph{The construction.} $M$'s finite control comprises $(q_1,q_2,\beta)$ where $q_j\in Q_j$ is the current state of~$M_j$ and $\beta$~is a displacement buffer holding at most~$2k$ tagged symbols. The stack alphabet consists of tagged pairs $(j,\gamma)$ with $j\in\{1,2\}$ and $\gamma\in\Gamma_1\cup\Gamma_2$.

At each input position, $M$ simulates both~$M_1$ and~$M_2$ on their respective state components:
\begin{itemize}
\item \emph{Pushes:} When $M_j$ pushes~$\gamma$, push $(j,\gamma)$ directly onto~$M$'s stack.
\item \emph{Pops:} When $M_j$ wants to pop expected symbol~$\gamma$: if the stack top is $(j,\gamma)$, pop it. Otherwise, pop entries from the stack into the displacement buffer~$\beta$ (at most~$2k$ entries, all with the wrong owner tag) until $(j,\gamma)$ is found. Pop $(j,\gamma)$, then push all entries from~$\beta$ back to the stack in their original order. If $(j,\gamma)$ is not found within~$2k$ entries, this computation path rejects.
\end{itemize}

\Cref{fig:buffer} illustrates the displacement mechanism.

\begin{figure}[ht]
\centering
\resizebox{\textwidth}{!}{%
\begin{tikzpicture}[>=stealth,
  cell/.style={minimum width=50pt, minimum height=18pt, inner sep=0pt,
               font=\scriptsize, draw, thick},
  c1/.style={cell, fill=blue!10, draw=blue!40},
  c2/.style={cell, fill=red!8, draw=red!40},
  ct/.style={cell, fill=green!12, draw=green!60!black, very thick}]

  %% ============ LEFT PANEL ============
  \node[font=\small\bfseries] at (-1.5,5.9) {Before $M_1$ pops};

  % Stack entries (bottom to top, anchored)
  \node[c1] (L1) at (-1.5,0.0) {$(1,\alpha)$};
  \node[c2] (L2) at (-1.5,0.9) {$(2,\delta)$};
  \node[ct] (L3) at (-1.5,1.8) {$(1,\gamma)$};
  \node[c2] (L4) at (-1.5,2.7) {$(2,\epsilon)$};
  \node[c2] (L5) at (-1.5,3.6) {$(2,\zeta)$};
  \node[c2] (L6) at (-1.5,4.5) {$(2,\theta)$};

  % Labels — right side only, close to boxes
  \node[font=\scriptsize, anchor=west] at (-0.55,4.5) {$\longleftarrow$ top};
  \node[font=\scriptsize, green!60!black, anchor=west] at (-0.55,1.8) {$\longleftarrow$ target};

  % Bracket: simple vertical line + ticks for M2 group
  \draw[red!50, very thick] (-3.1,2.25) -- (-3.1,4.95);
  \draw[red!50, very thick] (-3.1,2.25) -- (-2.9,2.25);
  \draw[red!50, very thick] (-3.1,4.95) -- (-2.9,4.95);
  \node[font=\scriptsize, red!60, anchor=east] at (-3.25,3.6)
    {$\le 2k$ $M_2$ entries};

  %% ============ ARROW ============
  \draw[->, very thick, gray!50] (1.8,2.3) -- node[above,font=\footnotesize]{displace}
    node[below,font=\footnotesize]{and pop} (3.8,2.3);

  %% ============ RIGHT PANEL ============
  \node[font=\small\bfseries] at (6.5,5.9) {After displacement};

  % Remaining stack
  \node[c1] (R1) at (6.5,0.0) {$(1,\alpha)$};
  \node[c2] (R2) at (6.5,0.9) {$(2,\delta)$};
  \node[font=\scriptsize, anchor=west] at (7.45,0.9) {$\longleftarrow$ top};

  % --- Buffer label above ---
  \node[font=\footnotesize, orange!80!black] at (6.5,5.1) {buffer $\beta$ (finite control)};

  % Buffer — horizontal row
  \node[c2] (B1) at (5.0,4.2) {$(2,\theta)$};
  \node[c2] (B2) at (6.5,4.2) {$(2,\zeta)$};
  \node[c2] (B3) at (8.0,4.2) {$(2,\epsilon)$};

  % Buffer outline
  \draw[orange!70, very thick, rounded corners=3pt]
    ([xshift=-6pt,yshift=-6pt]B1.south west) rectangle ([xshift=6pt,yshift=6pt]B3.north east);

  % Popped entry
  \node[ct] (P) at (10.5,4.2) {$(1,\gamma)$};
  \draw[->, green!60!black, thick] (8.8,4.2) -- (P.west);
  \node[font=\scriptsize, green!60!black] at (10.5,3.5) {popped};

  % Restore arrow
  \draw[->, orange!60, thick] (6.5,3.5) -- (6.5,1.6);
  \node[font=\scriptsize, orange!70!black] at (7.3,2.6) {restore};

\end{tikzpicture}%
}% end resizebox
\caption{The displacement buffer mechanism (\Cref{thm:bounded-gap}). \emph{Left:} the combined stack interleaves entries from $M_1$ (blue) and $M_2$ (red). When $M_1$ pops, its target $(1,\gamma)$ (green) has at most~$2k$ crossing $M_2$ entries above it. \emph{Right:} the $M_2$ entries are displaced into the buffer~$\beta$ (finite control), the target is popped, and the displaced entries are restored.}
\label{fig:buffer}
\end{figure}
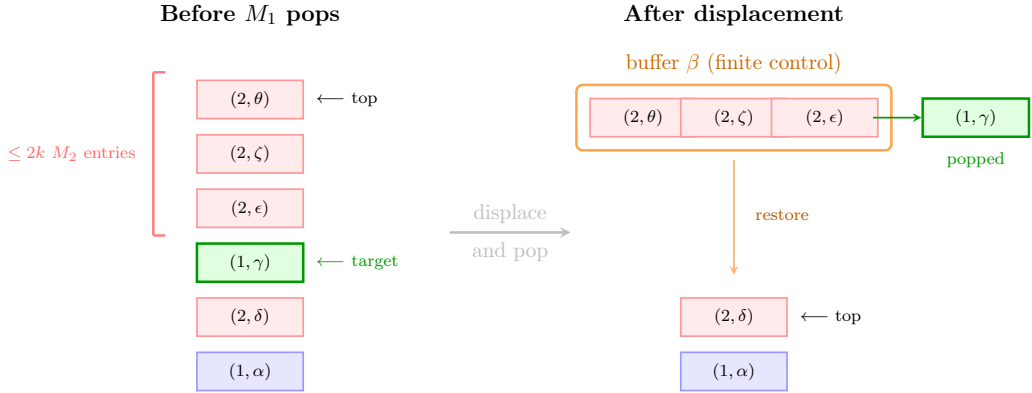

\emph{Displacement buffer bound.} Fix $w\in L(M_1)\cap L(M_2)$ and consider the accepting computation. We claim that when machine~$M_\ell$ ($\ell\in\{1,2\}$) pops arc $(i,j_{\mathrm{pop}})$, at most~$2k$ entries from the other machine~$M_{3-\ell}$ lie above~$M_\ell$'s target entry on the stack.

At time~$j_{\mathrm{pop}}$, an entry from~$M_{3-\ell}$ is above~$M_\ell$'s entry from position~$i$ iff it was pushed after position~$i$ and has not yet been popped. Such an $M_{3-\ell}$ arc~$(i',j')$ satisfies $i<i'$ and $j'>j_{\mathrm{pop}}$, which means $i<i'<j_{\mathrm{pop}}<j'$, that is, it \emph{crosses}~$(i,j_{\mathrm{pop}})$. By the $k$-bounded crossing gap hypothesis, $i'-i\le k$. Since each position contributes at most two pushes per machine (normal form), there are at most~$2k$ such entries.

(Entries from~$M_{3-\ell}$ arcs nested inside~$(i,j_{\mathrm{pop}})$, with $i<i'<j'<j_{\mathrm{pop}}$, have already been popped before time~$j_{\mathrm{pop}}$. Entries from~$M_{3-\ell}$ arcs pushed before~$i$ are below~$M_\ell$'s target. Entries from arcs not yet pushed are not on the stack.)

\emph{Stack order preservation.} The displacement operation pops~$\le 2k$ entries, removes the target, and pushes the displaced entries back in their original order. This preserves the relative ordering of all remaining stack entries.

\emph{Completeness.} For $w\in L(M_1)\cap L(M_2)$, the displacement buffer never exceeds~$2k$ entries, so no computation path rejects due to buffer overflow. Both simulated machines find their expected symbols and reach accepting states.

\emph{Soundness.} If $M$ accepts~$w$ on any nondeterministic path, both simulated machines accepted~$w$ (every pop verified the correct symbol), so $w\in L(M_1)\cap L(M_2)$.

Therefore $L(M) = L(M_1)\cap L(M_2)$ is context-free, with $|Q_1|\cdot|Q_2|\cdot(|\Gamma_1|{+}|\Gamma_2|{+}1)^{2k}$ states.
\end{proof}

\subsection{Discussion}

\Cref{thm:bounded-gap} resolves the bounded-gap case of the crossing-arc characterization: all pushes go directly to the stack, and at pop time at most~$2k$ entries from the other machine are displaced and restored. Combined with \Cref{thm:crossing}, the two resolved cases are $O(1)$ gap (CFL) and $\Theta(n)$ gap with crossing linkages (not CFL). In the next section, we introduce the \emph{inner segment measure} $\max(|P_2|,|P_3|)$ and show it is the quantity governing context-freeness, narrowing the open regime.

The interleaved palindrome (\Cref{ex:palindrome}) has crossing gap~1. \Cref{thm:bounded-gap} with $k=1$ constructs a PDA with displacement buffer of size~$2$: at each pop, at most two entries from the other machine must be displaced. The pair-alphabet construction of \Cref{ex:palindrome} is a more direct approach exploiting the specific structure of the interleaved palindrome language.

%% ====================================================================
\section{The inner segment dichotomy}\label{sec:intermediate}
%% ====================================================================

\Cref{thm:crossing} establishes that crossing pump-sensitive linkages on four segments of size~$\Omega(n)$ imply non-CFL. \Cref{thm:bounded-gap} establishes that crossing gap~$O(1)$ always admits a PDA construction. We now narrow the intermediate regime by proving a strengthened crossing theorem that replaces the ``all four segments large'' condition with a requirement on only the two \emph{inner} segments, and requires only \emph{one} of the two inner segments to be large.

\subsection{The strengthened crossing theorem}

\begin{theorem}[Strengthened crossing linkage theorem]\label{thm:strengthened}
Let $L\subseteq\Sigma^*$ contain, for infinitely many~$n$, a string $w_n = P_1 P_2 P_3 P_4$ satisfying:
\begin{enumerate}
\item[\textup{(i)}] $L$ has a pump-sensitive linkage $(P_1,P_3)$;
\item[\textup{(ii)}] $L$ has a pump-sensitive linkage $(P_2,P_4)$;
\item[\textup{(iii)}] $|P_2|\ge 1$, $|P_3|\ge 1$, and $\max(|P_2|,|P_3|)\ge n$.
\end{enumerate}
Then $L$ is not context-free.
\end{theorem}

Note the difference from \Cref{thm:crossing}: condition~(iii) requires only that the \emph{larger} of the two inner segments $P_2$ and~$P_3$ grows, while $P_1$, $P_4$, and the \emph{smaller} inner segment may be of any size, including~$O(1)$ (but not zero: the requirement $|P_2|\ge 1$, $|P_3|\ge 1$ ensures both linkages are non-vacuous; a linkage on an empty segment is always satisfied vacuously and provides no constraint). In the crossing-arc application, $|P_2|=i'-i\ge 1$ and $|P_3|=j-i'\ge 1$ follow from $i<i'<j$.

\begin{proof}
Assume for contradiction that $L$ is context-free with pumping length~$p$. Choose $n>p$ and consider $w_n = P_1 P_2 P_3 P_4$ with $\max(|P_2|,|P_3|)\ge n > p$.

By the pumping lemma, there exists a factorization $w_n = uvxyz$ with $|vxy|\le p$, $|vy|\ge 1$, and $uv^ixy^iz\in L$ for all $i\ge 0$. We consider two cases.

\medskip
\emph{Case A: $|P_2|\ge n > p$.} We enumerate all possible locations of~$vxy$:

\emph{Case~1: $vxy\subseteq P_1$.} Disjoint from~$P_3$ (separated by~$P_2$ of length $> p$). By~(i), contradiction.

\emph{Case~2: $vxy$ straddles the $P_1 P_2$ boundary.} Since $|vxy|\le p < |P_2|$, $vxy$ cannot reach past~$P_2$, so it is disjoint from~$P_3$ and~$P_4$. Intersects~$P_1$ but not~$P_3$; by~(i), contradiction.

\emph{Case~3: $vxy\subseteq P_2$.} Disjoint from~$P_4$ (separated by~$P_3$). By~(ii), contradiction.

\emph{Case~4: $vxy$ straddles the $P_2 P_3$ boundary.} Since $|vxy|\le p < |P_2|$, the start of~$vxy$ lies strictly inside~$P_2$, so $vxy$ is disjoint from~$P_1$. Two sub-cases: if $|P_3| \ge p$, then the end of~$vxy$ lies strictly inside~$P_3$, so $vxy$ is disjoint from~$P_4$; it intersects~$P_2$ but not~$P_4$, and (ii) gives a contradiction. If $|P_3| < p$, then $vxy$ may extend through~$P_3$ into~$P_4$; but it still intersects~$P_3$ while remaining disjoint from~$P_1$ (since $|P_2| > p$), so (i) gives a contradiction.

\emph{Case~5: $vxy\subseteq P_3$.} Disjoint from~$P_1$ (separated by~$P_2$ of length $> p$). By~(i), contradiction.

\emph{Case~6: $vxy$ straddles the $P_3 P_4$ boundary.} Disjoint from~$P_1$ (since $|P_2| > p$ separates $P_1$ from~$P_3$). Intersects~$P_3$ but not~$P_1$; by~(i), contradiction.

\emph{Case~7: $vxy\subseteq P_4$.} Disjoint from~$P_2$ (separated by~$P_3$). By~(ii), contradiction.

\medskip
\emph{Case B: $|P_3|\ge n > p$ (and $|P_2|$ may be small).} Cases~3 and~7 carry over from Case~A with~$P_3$ (of length~$>p$) providing the separation between~$P_2$ and~$P_4$. Cases~5 and~6 carry over because $vxy$ starts inside~$P_3$, hence is disjoint from~$P_1$ regardless of~$|P_2|$; by~(i), contradiction. The cases that differ are~1,~2, and~4, where a short~$P_2$ may let $vxy$ straddle across multiple segments:

\emph{Case~1: $vxy\subseteq P_1$.} Entirely within~$P_1$, hence disjoint from~$P_3$. By~(i), contradiction.

\emph{Case~2: $vxy$ straddles the $P_1 P_2$ boundary.} If $|P_2| \ge p$, the argument from Case~A applies. If $|P_2| < p$, then $vxy$ may extend through all of~$P_2$ into~$P_3$. Since $|P_3|>p\ge|vxy|$, $vxy$ cannot reach~$P_4$. Hence $vxy$ intersects~$P_2$ but is disjoint from~$P_4$; by~(ii), contradiction.

\emph{Case~4: $vxy$ straddles the $P_2 P_3$ boundary.} Since $|P_3|>p\ge|vxy|$, $vxy$ is disjoint from~$P_4$. If $vxy$ is also disjoint from~$P_1$, it intersects~$P_2$ but not~$P_4$; by~(ii), contradiction. If $|P_2|<p$ and $vxy$ extends back into~$P_1$, then $vxy$ straddles $P_1 P_2 P_3$ but is still disjoint from~$P_4$; it intersects~$P_2$ but not~$P_4$, so~(ii) again gives a contradiction.

\medskip
In both cases, all locations of~$vxy$ yield contradictions, so $L$ is not context-free.
\end{proof}

\begin{remark}[Three-segment straddles and the universality of linkages]\label{rem:straddle}
A subtle point arises when one inner segment (say~$P_2$) is small: a pumping-lemma factorization $vxy$ can straddle from~$P_1$ through~$P_2$ into~$P_3$. In this configuration, $vxy$ intersects \emph{both}~$P_1$ and~$P_3$, so linkage~(i) does not apply. However, $vxy$ also intersects~$P_2$ while remaining disjoint from~$P_4$ (since the other inner segment~$P_3$ is large), so linkage~(ii) resolves the case. This is why the universal quantification in \Cref{def:linkage} is essential: the linkage must hold for \emph{all} factorizations intersecting one linked segment, including cross-segment factorizations that also touch other segments. A language like $\{a^n b c^n d\}$, where pumping $a^k b c^k$ preserves membership, does \emph{not} satisfy linkage~$(P_2,P_4)$ under this universal definition, even though pumping the singleton~$b$ alone does break membership. Consequently, such languages do not satisfy the theorem's hypotheses.
\end{remark}

\subsection{Application to CFL intersection}

\begin{corollary}\label{cor:strengthened-cfl}
Let $L_1,L_2$ be context-free languages. Suppose there exist, for infinitely many~$n$, a string $w_n\in L_1\cap L_2$ and a factorization $w_n = P_1 P_2 P_3 P_4$ with $|P_2|\ge 1$, $|P_3|\ge 1$, and $\max(|P_2|,|P_3|)\ge n$, such that $L_1$ has a pump-sensitive linkage $(P_1,P_3)$ and $L_2$ has a pump-sensitive linkage $(P_2,P_4)$. Then $L_1\cap L_2$ is not context-free.
\end{corollary}

\begin{proof}
Pump-sensitive linkages are inherited by $L_1\cap L_2\subseteq L_i$. Apply \Cref{thm:strengthened}.
\end{proof}

\subsection{Interpretation in terms of crossing geometry}\label{sec:crossing-geometry}

In the crossing-gap setup, arcs $(i,j)\in\mu_{M_1}(w)$ and $(i',j')\in\mu_{M_2}(w)$ with $i<i'<j<j'$ define three segments between the crossing endpoints: the \emph{left inner segment} $w[i{+}1..i']$ (of length~$i'-i$), the \emph{right inner segment} $w[i'{+}1..j]$ (of length~$j-i'$), and the \emph{right gap} $w[j{+}1..j']$ (of length~$j'-j$). The string also has a prefix $w[1..i]$ and a suffix $w[j'{+}1..|w|]$. When applying the abstract factorization theorems (\Cref{thm:crossing,thm:strengthened}, or \Cref{cor:strengthened-cfl} when linkages hold on the ambient languages), we set $P_1 = w[1..i]$ (the prefix), $P_2 = w[i{+}1..i']$ (left inner segment), $P_3 = w[i'{+}1..j]$ (right inner segment), and $P_4 = w[j{+}1..|w|]$ (right gap and suffix combined), obtaining a complete factorization $w = P_1 P_2 P_3 P_4$. The inner segments $P_2$ and~$P_3$, which govern the dichotomy, correspond exactly to the geometric segments between the interleaved crossing endpoints.

Two measures on the crossing endpoints are natural:
\begin{itemize}
\item The \emph{crossing gap}: $\max(i'-i,\;j'-j)$, the larger of the push-endpoint distance and the pop-endpoint distance.
\item The \emph{inner segment measure}: $\max(|P_2|,|P_3|) = \max(i'-i,\;j-i')$, the maximum inner segment length.
\end{itemize}

These measures differ because the right gap $j'-j$ is independent of the inner segments. The inner segment measure governs the span of the arc $(i,j)$, which we call the \emph{left-starting arc} (the one whose push endpoint $i$ precedes the other arc's push endpoint $i'$). Its span is $j - i = (i'-i) + (j-i') = |P_2| + |P_3|$. When $\max(|P_2|,|P_3|) \le D$, this arc has span at most~$2D$, making it a ``short arc.''

\Cref{thm:strengthened} shows: if $\max(|P_2|,|P_3|) \to \infty$ for some family of crossing pairs with pump-sensitive linkages, then $L_1 \cap L_2$ is not CFL. The previous \Cref{thm:crossing} required all four segments to satisfy $|P_i| \ge n$; the strengthening to a single inner segment growing is strictly more general.

\subsection{The inner segment dichotomy}

\begin{definition}[Inner segment measure]\label{def:inner-segment}
For a crossing pair of arcs from different machines, written with $i<i'<j<j'$, the \emph{inner segment measure} is $\max(|P_2|,|P_3|) = \max(i'-i,\;j-i')$.
\end{definition}

When the inner segment measure $\max(i'-i,\;j-i') = O(1)$ while the crossing gap $\max(i'-i,\;j'-j) = \omega(1)$, we informally call the crossing \emph{near-tangent}: both inner endpoints $i'$ and~$j$ are within bounded distance of the left-starting arc's endpoints, but the right-starting arc's pop at~$j'$ is far away. The unbounded gap arises entirely from the right gap~$j'-j$, while both inner segments remain bounded. This is exactly the bounded-inner-segment, unbounded-gap case resolved by \Cref{thm:buffer}.

\Cref{fig:arcs} illustrates the two cases of the inner segment dichotomy.

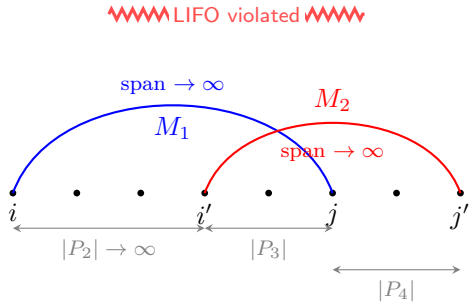
\begin{figure}[H]
\centering
\begin{tikzpicture}[>=stealth, scale=0.75]
  % === Left panel: Bounded inner segment (CFL) ===
  \node[anchor=south] at (4.5,3.8) {\textbf{Bounded inner segments: CFL}};
  % String positions
  \foreach \x in {0,1,...,9} \fill (\x,0) circle (1.5pt);
  \node[below,font=\footnotesize] at (0,0) {$i$};
  \node[below,font=\footnotesize] at (2,0) {$i'$};
  \node[below,font=\footnotesize] at (3,0) {$j$};
  \node[below,font=\footnotesize] at (9,0) {$j'$};
  % M1 arc (short) — span label above, machine label below
  \draw[thick, blue] (0,0) to[out=80,in=100]
    node[above,font=\scriptsize,pos=0.5,blue]{span $\le 2D$: buffer}
    node[below,font=\footnotesize,pos=0.5,blue]{$M_1$}
    (3,0);
  % M2 arc (long) — span label above, machine label below
  \draw[thick, red] (2,0) to[out=65,in=115]
    node[above,font=\scriptsize,pos=0.5,red]{span $\omega(1)$: stack}
    node[below,font=\footnotesize,pos=0.5,red]{$M_2$}
    (9,0);
  % Segment labels
  \draw[<->,gray,font=\scriptsize] (0,-0.55) -- node[below]{$|P_2|\le D$} (2,-0.55);
  \draw[<->,gray,font=\scriptsize] (2,-0.55) -- node[below]{$|P_3|\le D$} (3,-0.55);
  \draw[<->,gray,font=\scriptsize] (3,-1.2) -- node[below]{$|P_4| = \omega(1)$} (9,-1.2);
\end{tikzpicture}

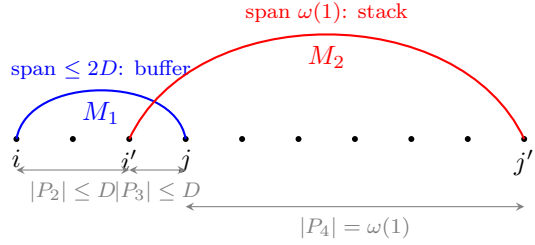
\medskip

\begin{tikzpicture}[>=stealth, scale=0.85]
  % === Bottom panel: Growing inner segment (not CFL) ===
  \node[anchor=south] at (3.5,4.2) {\textbf{Growing inner segments: not CFL}};
  % String positions
  \foreach \x in {0,1,...,7} \fill (\x,0) circle (1.5pt);
  \node[below,font=\footnotesize] at (0,0) {$i$};
  \node[below,font=\footnotesize] at (3,0) {$i'$};
  \node[below,font=\footnotesize] at (5,0) {$j$};
  \node[below,font=\footnotesize] at (7,0) {$j'$};
  % M1 arc — span label above, machine label below
  \draw[thick, blue] (0,0) to[out=70,in=110]
    node[above,font=\scriptsize,pos=0.5,blue]{span $\to\infty$}
    node[below,font=\footnotesize,pos=0.5,blue]{$M_1$}
    (5,0);
  % M2 arc — machine label above
  \draw[thick, red] (3,0) to[out=70,in=110]
    node[above,font=\footnotesize,pos=0.5,red]{$M_2$}
    (7,0);
  % Red span label below M2 arc, centered at midpoint x=5
  \node[red,font=\scriptsize] at (5,0.6) {span $\to\infty$};
  % LIFO failure indicator: jagged line just above arc peaks
  \draw[red!70,very thick,decorate,decoration={zigzag,segment length=4pt,amplitude=2pt}]
    (1.5,2.8) -- (5.5,2.8);
  \node[red!70,font=\scriptsize\sffamily,fill=white,inner sep=1.5pt] at (3.5,2.8) {LIFO violated};
  % Segment labels
  \draw[<->,gray,font=\scriptsize] (0,-0.55) -- node[below]{$|P_2|\to\infty$} (3,-0.55);
  \draw[<->,gray,font=\scriptsize] (3,-0.55) -- node[below]{$|P_3|$} (5,-0.55);
  \draw[<->,gray,font=\scriptsize] (5,-1.2) -- node[below]{$|P_4|$} (7,-1.2);
\end{tikzpicture}
\caption{The inner segment dichotomy. \emph{Top:} bounded inner segments ($\max(|P_2|,|P_3|)\le D$); the left-starting crossing arc has span $\le 2D$ and is absorbed into a finite buffer (\Cref{thm:buffer}). \emph{Bottom:} growing inner segments ($\max(|P_2|,|P_3|)\to\infty$); the left-starting arc grows unboundedly, the pumping lemma substring $vxy$ cannot straddle both linked segments, and the language is not CFL (\Cref{thm:strengthened}).}
\label{fig:arcs}
\end{figure}

The full classification is:

\begin{center}
\resizebox{\textwidth}{!}{%
\begin{tabular}{llll}
\toprule
Crossing type & Conditions & Intersection & Reference \\
\midrule
No crossings & -- & CFL & (product construction) \\
Bounded gap & $\max(i'\!-\!i,\,j'\!-\!j) = O(1)$ & CFL & \Cref{thm:bounded-gap} \\
Bounded inner segment & $\max(|P_2|,|P_3|)=O(1)$, gap $\omega(1)$ & CFL & \Cref{thm:buffer} \\
Growing inner seg.$^\dagger$ & $\max(|P_2|,|P_3|)\to\infty$, gap $\omega(1)$ & Not CFL & \Cref{thm:strengthened} \\
\bottomrule
\end{tabular}}
\end{center}
{\small $^\dagger$Requires pump-sensitive linkages on the crossing segments. For block-counting CFLs, these always hold on~$L_1\cap L_2$. For general CFLs, see \Cref{rem:exhaustiveness}. The gap restriction $\omega(1)$ follows from the bounded-gap row: bounded gap always yields CFL, so the linkage conditions are satisfiable only when the gap is also unbounded.}

The table reveals a previously unrecognized case: bounded inner segments with \emph{unbounded} crossing gap, yet the intersection remains CFL (\Cref{thm:buffer}). The remaining open question concerns the non-CFL direction: whether growing inner segments \emph{without} pump-sensitive linkages can occur. For block-counting CFLs, linkages are automatic; for general CFLs, this is the main open problem (\Cref{rem:exhaustiveness}).

\subsection{Refutation of the sharp gap conjecture}

One might conjecture that crossing gap $\omega(1)$ alone implies non-CFL. The inner segment dichotomy reduces this to a geometric question: does crossing gap $\omega(1)$ force some crossing pair to have $\max(|P_2|,|P_3|) = \omega(1)$? We show the answer is \emph{no}.

\begin{proposition}[Refutation of the sharp gap conjecture]\label{prop:gap-refuted}
There exist context-free languages $L_1,L_2$ with PDAs $M_1,M_2$ such that:
\begin{enumerate}
\item[\textup{(i)}] $L(M_1)\cap L(M_2)$ is context-free;
\item[\textup{(ii)}] for infinitely many~$n$, there exist crossing arcs with crossing gap exceeding~$n$;
\item[\textup{(iii)}] the inner segment measure satisfies $\max(|P_2|,|P_3|) = O(1)$ for all crossing pairs.
\end{enumerate}
\end{proposition}

\begin{proof}
Let $\Sigma = \{a,b,d,e,f,g,h\}$. Define:
\begin{align*}
L_1 &= \{a\,b\,a\,w_1\,f\,g^k h^k : w_1\in\{d,e\}^*,\; k\ge 0\}, \\
L_2 &= \{a\,b\,a\,d^n e^n f\,w_2 : w_2\in\{g,h\}^*,\; n\ge 0\}.
\end{align*}
Both are context-free. $M_1$ pushes at position~1 (reading~$a$), pops at position~3 (reading the second~$a$), reads $d$'s, $e$'s, $f$ freely, then pushes $g$'s and pops $h$'s. $M_2$ pushes at position~2 (reading~$b$), pushes $d$'s, pops $e$'s (nested inside the $b$-push), pops at position~$3{+}2n{+}1$ (reading~$f$), then reads $g$'s and $h$'s freely.

\emph{Intersection.} $L_1\cap L_2 = \{a\,b\,a\,d^n e^n f\,g^k h^k : n,k\ge 0\}$, which is context-free: a PDA reads prefix~$aba$ in finite state, pushes~$d$'s and pops~$e$'s, reads~$f$, then pushes~$g$'s and pops~$h$'s.

\emph{Crossing structure.} $M_1$'s arc $(1,3)$ and $M_2$'s arc $(2,\,3{+}2n{+}1)$ satisfy $1<2<3<3{+}2n{+}1$, giving a crossing pair with inner segments $|P_2| = i'-i = 2-1 = 1$ and $|P_3| = j-i' = 3-2 = 1$, and right gap $j'-j=2n{+}1$. The inner segment measure is $\max(|P_2|,|P_3|)=1$: the ``interruption'' of $M_2$'s arc by $M_1$'s push-pop pair lasts exactly one position in each direction. Meanwhile, the crossing gap $\max(i'-i,\;j'-j) = \max(1,2n{+}1)\to\infty$. No other crossing pairs exist: $M_1$'s long arcs (in the $g$-$h$ region) and $M_2$'s arcs (at positions $\le 3{+}2n{+}1$) occupy disjoint position ranges.

Thus $L_1\cap L_2$ is CFL, the crossing gap is~$\omega(1)$, and all crossing pairs have inner segment measure~$O(1)$.
\end{proof}

This works because $M_1$'s crossing arc has bounded span ($j-i=2$): it is a short interruption of $M_2$'s long arc. Such short arcs can be absorbed into a finite buffer; they behave like finite-state conditions. The long non-crossing arcs of both machines are well-nested and share the stack without conflict. We formalize this intuition in \Cref{thm:buffer} below.

\begin{remark}\label{rem:gap-vs-inner}
The refutation separates two measures that might seem equivalent. The crossing gap $\max(i'-i,\;j'-j)$ grows without bound, but the inner segment measure $\max(|P_2|,|P_3|)$ stays bounded: one arc is ``short'' (bounded span) despite crossing a ``long'' arc (unbounded span). Context-freeness is governed by the inner segment measure, not by the crossing gap.
\end{remark}

\subsection{Block-counting crossings always have growing inner segments}\label{sec:block-inner}

For block-counting CFLs with crossing arcs, the right inner segment is always $\Theta(n)$: each arc connects blocks $B_i$ and~$B_j$ with $|B_i|=|B_j|=\Theta(n)$, and any crossing pair has at least one full block between the inner endpoints. Consequently, all crossings in the block-counting setting have $\max(|P_2|,|P_3|) = \Theta(n)$. Moreover, pump-sensitive linkages hold on the intersection~$L_1\cap L_2$: in the crossing-arc geometry, each segment $P_i$ lies within blocks constrained by some equality, so any pump that modifies characters in one linked segment without its partner changes a constrained block count while leaving its partner's unchanged, violating one of the intersection's constraints (cf.\ \Cref{rem:linkage-ambient}). Thus \Cref{thm:strengthened} resolves the entire block-counting case, recovering \Cref{thm:characterization} as a consequence of the inner-segment-based analysis.

\subsection{The bounded inner segment resolution theorem}

When $\max(|P_2|,|P_3|) = O(1)$ for all crossing pairs, the left-starting arc in each pair has span~$O(1)$. The geometric meaning is this: the inner segment measure captures the length of the longer inner segment between the interleaved endpoints of two crossing arcs. A crossing arc~$(i,j)$ that is left-starting (i.e., $i$ is the smallest endpoint among the four) has its push at~$i$ and pop at~$j$ separated by at most~$2D$ positions; in between, the other machine may push and pop, but this ``interruption'' is bounded. Such short arcs impose only local constraints, akin to finite-state conditions, that a bounded buffer in the finite control can absorb, while all long arcs go directly to the stack (long arcs from different machines cannot cross, since any crossing forces one arc to be short). Bufferability is determined by interruption duration, not by the overall distance between the two machines' arcs.

\begin{theorem}[Bounded inner segment resolution]\label{thm:buffer}
Let $M_1,M_2$ be PDAs with state sets~$Q_1,Q_2$ and stack alphabets~$\Gamma_1,\Gamma_2$. If there exists a constant~$D$ such that for every $w\in L(M_1)\cap L(M_2)$ and every crossing pair with $i<i'<j<j'$, the inner segment measure satisfies $\max(i'-i,\;j-i')\le D$, then $L(M_1)\cap L(M_2)$ is recognized by a nondeterministic PDA with at most
\[
  |Q_1|\cdot|Q_2|\cdot\bigl(1+(|\Gamma_1|{+}|\Gamma_2|)\cdot 2D\bigr)^{8D}
\]
states. In particular, $L(M_1)\cap L(M_2)$ is context-free, with state complexity exponential in~$D$ but independent of the input length.
\end{theorem}

\begin{proof}
We construct a nondeterministic PDA~$M$ recognizing $L(M_1)\cap L(M_2)$. The proof proceeds in two stages: first, we define an \emph{oracle strategy} that classifies each arc as ``short'' or ``long'' using full knowledge of the accepting computations; then, we show that a nondeterministic PDA can guess this classification and simulate both machines correctly.

\emph{Stage~1: Oracle classification.} Fix $w\in L(M_1)\cap L(M_2)$ and fix accepting computations of~$M_1$ and~$M_2$ on~$w$, determining push-pop matchings $\mu_{M_1}(w)$ and $\mu_{M_2}(w)$. An arc $(i,j)\in\mu_{M_k}(w)$ is \emph{short} if $j - i \le 2D$; otherwise it is \emph{long}. (This classification depends only on the arc's own endpoints. Stage~2 shows a nondeterministic PDA can guess it; the oracle merely establishes that a correct guess \emph{exists}.)

\emph{Claim~A: In every crossing pair, at least one arc is short.} Suppose arcs $(i,j)$ and $(i',j')$ from different machines cross with $i<i'<j<j'$. By hypothesis, $\max(i'-i,\;j-i')\le D$, so $j - i = (i'-i)+(j-i') \le 2D$. That is, the arc~$(i,j)$ is short. (The other arc~$(i',j')$ may have unbounded span; $j' - i'$ is not constrained by~$D$.)

\emph{Claim~B: Long arcs are jointly well-nested.} Each machine's arcs are individually well-nested. We show that no long arc from~$M_1$ crosses any long arc from~$M_2$. If arcs $(a,b)\in\mu_{M_1}(w)$ and $(c,d)\in\mu_{M_2}(w)$ cross, then Claim~A guarantees one of them has span~$\le 2D$, i.e., one is short. Contrapositively, if both are long (span~$> 2D$), they do not cross. By \Cref{lem:lifo}, the long arcs from both machines form a jointly well-nested set.

\emph{Stage~2: PDA construction.} $M$'s finite control comprises the state pairs $(q_1,q_2)\in Q_1\times Q_2$ and a buffer~$B$ holding at most~$8D$ pending entries, each recording a stack symbol from~$\Gamma_1\cup\Gamma_2$ and a countdown timer in~$\{1,\dots,2D\}$. The buffer contributes a finite factor to the state count: at most $(1+(|\Gamma_1|{+}|\Gamma_2|)\cdot 2D)^{8D}$, since each of the~$8D$ slots is either empty or holds a symbol-timer pair.

$M$'s stack alphabet consists of tagged pairs $(\gamma,\mathrm{owner})$ with $\gamma\in\Gamma_1\cup\Gamma_2$ and $\mathrm{owner}\in\{1,2\}$.

At each input position, $M$ simulates both~$M_1$ and~$M_2$ on their respective state components. When a simulated machine performs a push, $M$ nondeterministically guesses whether the arc is short or long:
\begin{itemize}
\item \emph{Long guess:} push the symbol onto $M$'s stack, tagged with its owner.
\item \emph{Short guess:} place the symbol in the buffer with countdown $2D$.
\end{itemize}
At each step, all countdown timers decrement. A buffer entry whose timer reaches~$0$ without being popped causes this computation path to reject.

When a simulated machine performs a pop: if the expected symbol is in the buffer, resolve it there; if the top of stack has the correct tag and symbol, pop it; otherwise, reject this path.

\emph{Completeness (correct guesses accept).} Given $w\in L(M_1)\cap L(M_2)$, the oracle classification from Stage~1 defines a specific guess sequence. On this path:
\begin{itemize}
\item Every short arc (span $\le 2D$) is buffered and resolves before its timer expires.
\item Every long arc is pushed to the stack. By Claim~B, long arcs from different machines are jointly well-nested, so the LIFO discipline handles them correctly: nested arcs are pushed after and popped before the enclosing arc, maintaining stack order. Owner tags resolve which machine each symbol belongs to.
\item At most~$4D$ short arcs per machine are simultaneously pending (each spans $\le 2D$ positions, at most two pushes per position per machine), so the buffer never exceeds~$8D$ entries.
\end{itemize}
Both simulated machines reach accepting states and the buffer empties, so $M$ accepts.

\emph{Soundness.} If $M$ accepts~$w$ on any nondeterministic path, both simulated machines accepted~$w$, so $w\in L(M_1)\cap L(M_2)$.

Therefore $L(M_1)\cap L(M_2) = L(M)$ is context-free.
\end{proof}

Combining \Cref{thm:strengthened} and \Cref{thm:buffer}:

\begin{corollary}[Inner segment dichotomy]\label{cor:dichotomy}
Let $M_1,M_2$ be PDAs. Let $L = L(M_1)\cap L(M_2)$.
\begin{enumerate}
\item[\textup{(a)}] If there exists a constant~$D$ such that every crossing pair on every $w\in L$ has inner segment measure $\max(|P_2|,|P_3|)\le D$, then $L$ is context-free \textup{(\Cref{thm:buffer})}.
\item[\textup{(b)}] If for infinitely many~$n$, some $w_n\in L$ has a crossing pair with pump-sensitive linkages and $\max(|P_2|,|P_3|)\ge n$, then $L$ is not context-free \textup{(\Cref{thm:strengthened})}.
\end{enumerate}
We call this a ``dichotomy'' in the sense that the two directions identify the governing quantity (inner segment measure rather than crossing gap). The term is exact for block-counting CFLs, where conditions~\textup{(a)} and~\textup{(b)} are exhaustive. For general CFLs, it is conditional: the gap between~\textup{(a)} and~\textup{(b)} concerns whether growing inner segments always carry pump-sensitive linkages (see \Cref{rem:exhaustiveness}).
\end{corollary}

\begin{remark}[Exhaustiveness]\label{rem:exhaustiveness}
For block-counting CFLs, conditions~\textup{(a)} and~\textup{(b)} are exhaustive: the rigid block structure guarantees that every crossing pair with growing inner segments carries pump-sensitive linkages (\Cref{sec:block-inner}), so the dichotomy is complete, recovering \Cref{thm:characterization}.

For general CFLs, exhaustiveness is open. The gap is the case where inner segments grow but no crossing pair has pump-sensitive linkages; alternative accepting computations in~$L_1$ or~$L_2$ might absorb pumps that would otherwise exit~$L$. Direction~\textup{(a)} is unconditional and requires no linkage assumptions. Direction~\textup{(b)} requires that the specific segment decomposition induced by~$M_1,M_2$'s matchings produces linkages on~$L = L_1\cap L_2$. Whether a non-CFL intersection always admits such a decomposition, equivalently, whether there exist PDAs whose crossing structure witnesses the non-CFL-ness via pump-sensitive linkages, remains open (\Cref{sec:open}).

We note that no example of a non-CFL intersection \emph{lacking} pump-sensitive linkages is currently known. Every non-CFL intersection in the literature (block-counting languages, Shieber's Swiss German construction, positional matching languages) exhibits pump-sensitive linkages on its crossing segments. A counterexample would require two CFLs whose intersection is non-CFL, yet where for every crossing pair with growing inner segments, some factorization touching one linked segment alone still pumps within~$L_1\cap L_2$. This would mean the crossing segments have enough ``slack'' that local modifications can be absorbed, despite the global language being non-CFL. Constructing such an example appears difficult: most natural CFL intersection constructions use equality constraints that create tight segment coupling.
\end{remark}

\subsection{Comparison with the bounded-gap construction}

The bounded-gap construction (\Cref{thm:bounded-gap}) pushes all symbols directly to the stack and displaces at most~$2k$ entries at pop time. When the crossing gap is bounded, the displacement count stays bounded. When the crossing gap is unbounded, displacement may grow without bound, exceeding any finite buffer. \Cref{thm:buffer} succeeds by a different strategy: instead of displacing at pop time, it classifies each arc as short (span~$\le 2D$) or long at push time. Short arcs (which include the left-starting arc of every crossing pair, by Claim~A) are buffered entirely in finite control; long arcs are pushed to the stack. Claim~B guarantees the stacked long arcs are jointly well-nested, since any crossing pair forces at least one arc to be short. The nondeterministic PDA guesses the short/long classification without needing to see the future.

%% ====================================================================
\section{Related work}\label{sec:related}
%% ====================================================================

\subsection{State complexity of automata intersection}

Yu, Zhuang, and Salomaa~\cite{yu1994} establish that the intersection of two DFAs with $m$ and $n$ states requires at most $mn$ states, tight in the worst case. For unary alphabets, the state complexity is $\operatorname{lcm}(m,n)$ when $m$ and~$n$ are coprime, which connects to the Chinese Remainder Theorem. Dutta and Saikia~\cite{dutta2022} have also connected the CRT to finite automata product constructions. Our \Cref{thm:buffer} extends this line of work to the pushdown case, where the intersection may leave the CFL class entirely.

\subsection{Intersection of pushdown automata}

The intersection non-emptiness problem for two CFLs is undecidable~\cite{hopcroft2007}. Swernofsky and Wehar~\cite{wehar2015} establish tight complexity bounds for deciding non-emptiness of intersections involving regular and context-free languages, connecting these problems to graph-theoretic hardness. Our results complement these decidability and complexity bounds by characterizing the \emph{language class} of two-CFL intersections: we identify geometric conditions on the crossing arcs that determine whether the intersection remains context-free.

\subsection{Mildly context-sensitive languages and LCFRS}

Vijay-Shanker, Weir, and Joshi~\cite{vijayshanker1987} introduced linear context-free rewriting systems (LCFRS), which generalize CFGs by allowing nonterminals to derive tuples of strings. An LCFRS of fan-out~$k$ generates exactly the $k$-multiple context-free languages ($k$-MCFL); fan-out~1 gives the CFLs. Seki, Matsumura, Fujii, and Kasami~\cite{seki1991} independently developed multiple context-free grammars (MCFGs), equivalent to LCFRS.

Many non-CFL intersections of two CFLs land in LCFRS(2). For instance, $\{a^n b^n c^n d^n\}$ is LCFRS(2) but not CFL; it arises as the intersection of two CFLs with crossing arcs. Whether two-CFL intersections with a bounded number of independent crossing families always lie in LCFRS($k$) for bounded~$k$ is an open problem (\Cref{sec:open}, Question~7).

Our inner segment measure connects naturally to LCFRS fan-out. When $\max(|P_2|,|P_3|)\le D$, the left-starting arc in each crossing pair has bounded span ($\le 2D$), and the buffer construction (\Cref{thm:buffer}) shows the intersection is CFL (fan-out~1). When inner segments grow, each crossing pair contributes an additional ``thread'' of dependency, suggesting the non-CFL intersection has fan-out~2.

\subsection{Crossing dependencies and well-nestedness}

Shieber~\cite{shieber1985} proved Swiss German is not context-free via cross-serial dependencies reducible to $\{a^m b^n c^m d^n\}$. Joshi~\cite{joshi1990} showed crossing dependencies require mildly context-sensitive formalisms. Bodirsky, Kuhlmann, and M\"ohl~\cite{bodirsky2005} formalized \emph{well-nestedness}, and Kuhlmann~\cite{kuhlmann2013} established a hierarchy in the LCFRS framework. Our \Cref{thm:characterization} can be seen as an intersection-theoretic analogue: where their work characterizes which single grammars stay within CFL, our result characterizes which pairs of CFL constraints stay within CFL under intersection.

%% ====================================================================
\section{Discussion}\label{sec:discussion}
%% ====================================================================

\subsection{Position in the Chomsky hierarchy}

The inner segment dichotomy places CFL intersection within a graded hierarchy. For finite automata, the intersection of two regular languages is always regular, with $mn$ product states in the worst case~\cite{yu1994}. For pushdown automata, our results show the intersection may leave the CFL class entirely. Combined with the LCFRS framework~\cite{vijayshanker1987,seki1991}:

\begin{center}
\resizebox{\textwidth}{!}{%
\begin{tabular}{clll}
\toprule
Fan-out & Class & Intersection behavior & Reference \\
\midrule
-- & Regular & Always closed; $mn$ product states (unary: $\operatorname{lcm}$) & \cite{yu1994} \\
1 & CFL & Closed if bounded gap, bounded inner seg., or jointly well-nested$^\ddagger$ & \Cref{thm:bounded-gap,thm:buffer,thm:characterization} \\
2 & LCFRS(2) & Non-CFL with crossing arcs (conjectured upper bound) & \Cref{thm:strengthened}; Q.\,7 \\
\bottomrule
\end{tabular}}
\end{center}
{\small $^\ddagger$``Iff'' for block-counting CFLs (\Cref{thm:characterization}); ``if'' for general CFLs (non-CFL direction requires pump-sensitive linkages; see \Cref{rem:exhaustiveness}).}

\medskip

For $k\ge 3$ CFLs, the situation is qualitatively different. The intersection of any finite number of CFLs remains context-sensitive (hence decidable), but every recursively enumerable language is the homomorphic image of the intersection of two linear CFLs~\cite{baker1974}: the homomorphism, not the intersection itself, introduces full Turing power.

The inner segment dichotomy (\Cref{cor:dichotomy}) characterizes a sufficient condition for crossing the CFL/LCFRS(2) boundary: growing inner segments (with pump-sensitive linkages) push the intersection to LCFRS(2) or beyond, where a single stack no longer suffices and multiple threads of dependency are required.

\subsection{Geometry vs.\ complexity}

The non-CFL results concern \emph{memory access discipline}, not memory amount. For the block-counting CFLs central to our examples, both the individual languages and their non-CFL intersection (e.g., $\{a^n b^n c^n d^n\}$) can be verified in $O(\log n)$ space on a Turing machine by counting and comparing block lengths. The distinction is that individual CFL regularities fit on a stack (LIFO access), while the joint regularity requires non-LIFO access. The inner segment measure captures precisely this: when some inner segment grows unboundedly, the two machines' stack operations interleave in a way that no single stack can simulate.

%% ====================================================================
\section{Open questions}\label{sec:open}
%% ====================================================================

\begin{enumerate}

\item \textbf{Necessity of pump-sensitive linkages.} The inner segment dichotomy (\Cref{cor:dichotomy}) is complete for block-counting CFLs but conditional for general CFLs: direction~(b) requires pump-sensitive linkages on the crossing segments. If $L_1\cap L_2$ is not context-free and has crossing arcs with growing inner segments, must some PDA pair $(M_1,M_2)$ for $(L_1,L_2)$ produce crossing segments that carry pump-sensitive linkages on~$L_1\cap L_2$? A positive answer would make the dichotomy unconditionally complete. A negative answer would exhibit a non-CFL intersection where no machine-induced decomposition yields the linkage condition, requiring a fundamentally different proof technique (e.g., Ogden's lemma or the interchange lemma) for the non-CFL direction.

\item \textbf{Grammar-independent matching arcs.} The matching arcs $\mu_M(w)$ depend on the specific PDA~$M$ and its accepting computation, not on the language alone. Two PDAs for the same CFL can produce different matchings and different crossing structures. Visibly pushdown languages (VPLs) provide a natural grammar-independent notion of matching arcs via a fixed partition $\Sigma = \Sigma_c \cup \Sigma_r \cup \Sigma_\ell$. What is the intersection theory of VPLs over different partitions $(\Sigma_c, \Sigma_r, \Sigma_\ell)$ and $(\Sigma'_c, \Sigma'_r, \Sigma'_\ell)$? This may be the most tractable path to a fully grammar-independent characterization, and would also resolve question~1 in the VPL setting (since VPL matchings are determined by the input string).

\item \textbf{Invariance of crossing measures under PDA conversion.} The inner segment bound~$D$ in \Cref{thm:buffer} depends on the specific PDAs $M_1,M_2$ and their GNF normal forms, not just on~$L_1$ and~$L_2$. Is the property ``there exist PDAs $M_1$ for~$L_1$ and $M_2$ for~$L_2$ such that all crossing pairs have $\max(|P_2|,|P_3|)\le D$'' equivalent to ``$L_1\cap L_2$ is context-free''? More generally, does every PDA pair for~$(L_1,L_2)$ yield the same bounded/unbounded classification of inner segments, or can a different PDA choice change the crossing structure qualitatively?

\item \textbf{State complexity of the buffer construction.} \Cref{thm:buffer} constructs a PDA with $(1{+}(|\Gamma_1|{+}|\Gamma_2|)\cdot 2D)^{8D}\cdot|Q_1|\cdot|Q_2|$ states, exponential in~$D$. Is this blowup necessary, or can a polynomial-state PDA be found? More specifically, is the bound tight for any family of PDAs, and can the base or exponent be reduced?

\item \textbf{Descriptive complexity of CFL intersections.} Are there cases where two CFLs have a CFL intersection that requires a strictly larger grammar (more nonterminals or rules) than either component? This would provide a quantitative complexity measure within the CFL class.

\item \textbf{Indexed languages.} Some PDA intersection problems may produce languages at the indexed level~\cite{aho1968} rather than at the CFL or LCFRS level. Characterizing which crossing structures produce indexed languages vs.\ LCFRS languages is open.

\item \textbf{LCFRS fan-out of CFL intersections.} When $L_1\cap L_2$ is not CFL, it is often recognizable by an LCFRS of fan-out~2~\cite{seki1991,vijayshanker1987}. Is this always the case for two-CFL intersections? More precisely, does the number of independent crossing families in the matching graphs determine the LCFRS fan-out, and is it always bounded for the intersection of exactly two CFLs?

\item \textbf{Decidability of bounded inner segments.} Given two context-free grammars $G_1, G_2$, is it decidable whether there exist PDAs $M_1, M_2$ for $L(G_1), L(G_2)$ such that all crossing pairs have bounded inner segment measure? Since CFL intersection emptiness is undecidable in general, the answer is likely negative, but a proof or a reduction would clarify the computational boundary of the inner segment dichotomy.

\end{enumerate}

\section*{Acknowledgements}

\begin{sloppypar}
This work was partially funded by national funds through FCT\,--\,Foundation for Science and Technology, I.P., in the context of the project UIDB/00127/2020 (\url{https://doi.org/10.54499/UIDB/00127/2020}).
\end{sloppypar}

\section*{Declaration of generative AI use}

During the preparation of this work, the author used Claude (Anthropic) to assist with drafting the text and refining the structure of the mathematical arguments. The author verified all proofs manually and takes full responsibility for the content of the publication.

\section*{Credit authorship contribution statement}

\textbf{Jorge Miguel Silva:} Conceptualization, Methodology, Formal analysis, Writing -- original draft, Writing -- review \& editing.

\section*{Declaration of competing interest}

The author declares no competing interests.

\section*{Data availability}

No data were used for the research described in this article. All results are of a theoretical nature.

\bibliographystyle{elsarticle-harv}
\bibliography{refereces}
\end{document}